\documentclass[ejsv2,noshowframe]{imsart}

\RequirePackage[authoryear]{natbib}
\RequirePackage[colorlinks,citecolor=black,urlcolor=black]{hyperref}
\RequirePackage{graphicx}
\usepackage[normalem]{ulem}

\startlocaldefs
\theoremstyle{plain}

\newtheorem{theorem}{Theorem}[section]
\newtheorem{lemma}[theorem]{Lemma}
\theoremstyle{definition}

\theoremstyle{remark}

\usepackage{dsfont}
\usepackage{xcolor}
\usepackage{booktabs}
\usepackage{multirow}
\usepackage{listings}
\usepackage{placeins}

\newcommand{\pos}{x}
\newcommand{\latpos}{\tilde{x}}
\newcommand{\velo}{v}
\newcommand{\mo}{p}

\newcommand{\Tt}{T}

\newcommand{\thinnerspace}{\mskip.5\thinmuskip}

\newcommand{\hmc}{Hamiltonian Monte Carlo}
\newcommand{\pdmps}{piecewise-deterministic Markov processes}
\newcommand{\pdmp}{piecewise-deterministic Markov process}
\newcommand{\hzz}{Hamiltonian zig-zag sampler}
\newcommand{\zz}{zig-zag}

\newcommand{\pslab}{p_{\mathrm{slab}}}
\newcommand{\slab}{\pi_{\mathrm{slab}}}
\newcommand{\diff}{\operatorname{\mathrm{d}}\!{}}
\newcommand{\indicator}{\mathds{1}}
\newcommand{\La}{\text{\scriptsize L}}
\newcommand{\Or}{\text{\scriptsize O}}

\definecolor{emerald}{RGB}{12, 166, 151}
\definecolor{lava}{rgb}{0.81, 0.06, 0.13}

\newcommand{\var}{\operatorname{Var}}
\newcommand{\E}{\mathbb{E}}

\definecolor{jhublue}{HTML}{002D72}
\definecolor{jhured}{HTML}{CF4520}

\endlocaldefs

\begin{document}
\begin{frontmatter}
\title{Smoothing Out Sticking Points: Sampling from Discrete-Continuous Mixtures with Dynamical Monte Carlo by Mapping Discrete Mass into a Latent Universe}
\runtitle{Smoothing Out Sticking Points}

\begin{aug}
\author[A]{\fnms{Andrew}~\snm{Chin}\ead[label=e1]{achin23@jhu.edu}},
\and
\author[A]{\fnms{Akihiko}~\snm{Nishimura}\ead[label=e2]{aki.nishimura@jhu.edu}}
\address[A]{Department of Biostatistics, Bloomberg School of Public Health, Johns Hopkins University\printead[presep={,\ }]{e1,e2}}

\runauthor{A. Chin et al.}
\end{aug}

\begin{abstract}
Combining a continuous ``slab" density with discrete ``spike" mass at zero, spike-and-slab priors provide important tools for inducing sparsity and carrying out variable selection in Bayesian models.
However, the presence of discrete mass makes posterior inference challenging.
``Sticky" extensions to \pdmp{} samplers have shown promising performance, where sampling from the spike is achieved by the process sticking there for an exponentially distributed duration.
As it turns out, the sampler remains valid when the exponential sticking time is replaced with its expectation.
We justify this by mapping the spike to a continuous density over a latent universe, allowing the sampler to be reinterpreted as traversing this universe while being stuck in the original space.
This perspective opens up an array of possibilities to carry out posterior computation under spike-and-slab type priors. 
Notably, it enables us to construct sticky samplers using other dynamics-based paradigms such as \hmc{};
in fact, original sticky process can be established as a partial position-momentum refreshment limit of our Hamiltonian sticky sampler.
Our theoretical and empirical findings suggest these alternatives to be at least as efficient as the original sticky approach.
\end{abstract}


\begin{keyword}
\kwd{spike and slab}
\kwd{Bayesian variable selection}
\kwd{\pdmp{}}
\kwd{\hmc{}}
\end{keyword}

\end{frontmatter}

\section{Introduction}
In modern high dimensional problems, it is often desirable, or even necessary, to impose sparsity in models' parameters. 
In Bayesian paradigms, this is achieved through shrinkage priors with large probability near 0.
These priors fall into two main classes.
The first comprises spike-and-slab priors \citep{george1993variable, mitchell1988bayesian}, which mix a delta mass ``spike" at 0 and a continuous density ``slab" $\slab(\pos_i)$:
\begin{equation}\label{eq:ss_prior}
	\pi_0(\pos_i) = (1-\pslab)\delta_0(\pos_i) + \pslab \slab(\pos_i), \quad \pslab \in (0,1).
\end{equation}
The second comprises continuous shrinkage priors that concentrate probability mass near 0, but consist only of continuous densities.
While computationally convenient, the lack of delta mass means these priors cannot yield exact zeros in the posterior estimates, and so post-processing of the posterior is required to do variable selection \citep{hahn2015decoupling}. 
Overall, spike-and-slab priors remain an often preferred method for sparse estimation and yield optimal results in many cases \citep{carvalho2009handling, tadesse2021handbook}. 

Traditionally, spike-and-slab posteriors are sampled with Gibbs or reversible jump samplers, but both can mix poorly and the latter can be difficult to tune \citep{green2009reversible,o2009review}.
Recent intensive research on \pdmp{} samplers \citep{fearnhead2018piecewise, bierkens2019zig, bouchard2018bouncy} has given rise to novel ``sticky" samplers for these posteriors  \citep{bierkens2023sticky}. 
These samplers follow the standard \pdmp{} dynamics away from 0, but their coordinates ``stick" at 0 for exponentially distributed amounts of time.

\begin{figure}
	\centering
	\vspace{2mm}
	\includegraphics[width=8cm]{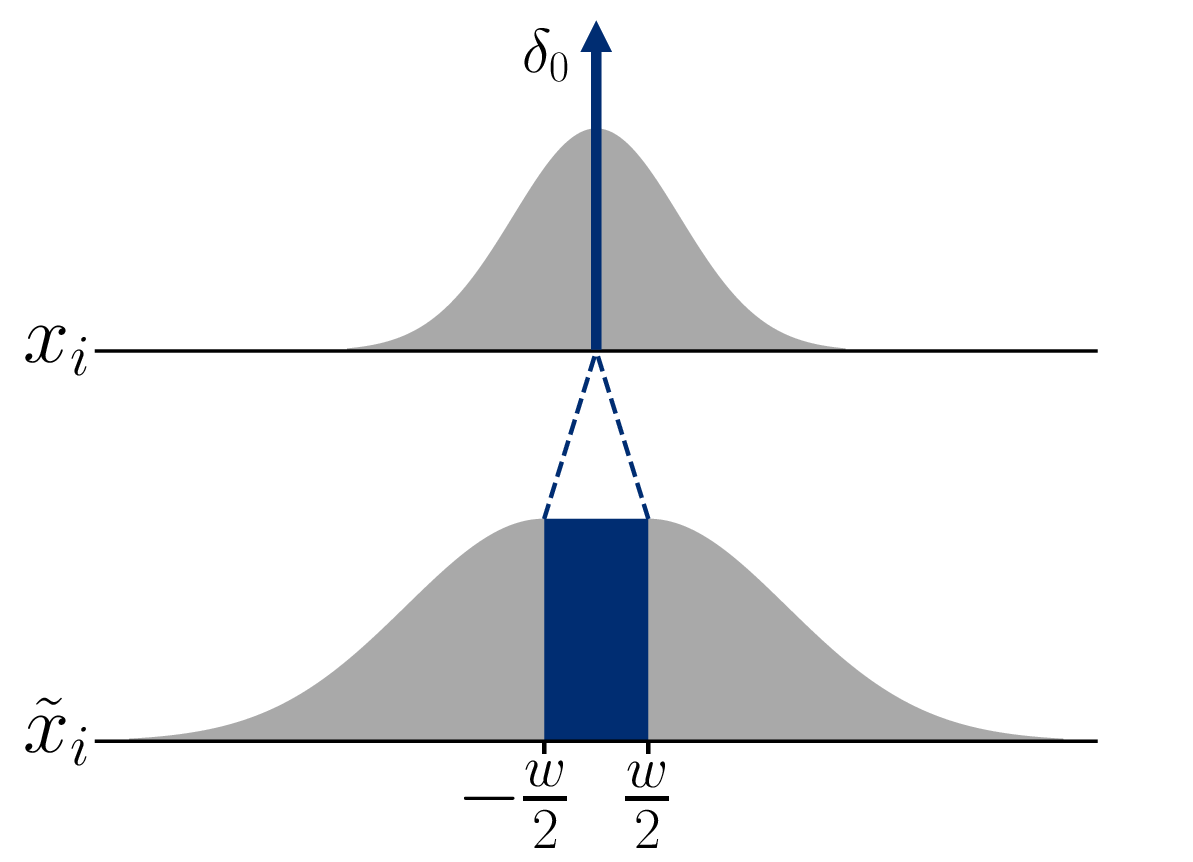}
	\caption{Constructing a continuous density representation of the spike-and-slab prior by spreading the spike mass over a latent universe and inserting it in the middle of the slab density.
		In other words, we introduce a latent parameter $\latpos_i$ with the continuous density which, when collapsing the latent universe to 0, 
		recovers the spike-and-slab prior on the original parameter $\pos_i$.
		A corresponding posterior on $\latpos_i$ also has a continuous density as long as the likelihood is a continuous function of $\pos_i$.
		Sampling from the discrete-mixture posterior thus reduces to first sampling from the latent continuous density and then mapping the samples back to the original space.%
	}
	\label{fig:augment}
\end{figure}

In this article, we observe that such sticking behavior can be viewed as the sampler entering a ``latent universe" when it reaches 0.
This universe is created by inserting the spike mass, as a continuous density over a latent space, in the middle of the slab density (Figure 1).
We can then sample from the resulting continuous density, map the latent universe back to the spike, and obtain draws from the original distribution.
The mapped trajectory of a continuous-time sampler now appears to stick at 0 while traveling through the latent universe.

Having reduced the task to sampling from a continuous density, we can straightforwardly apply existing algorithms such as \hmc{} and construct \textit{Hamiltonian sticky samplers}.
We can also run standard \pdmp{} samplers on the latent continuous density, yielding \textit{latent sticky samplers} with deterministic sticking times.
While embedding a discrete mass into continuous space is not a new concept \citep{petris2003geometric, pakman2013auxiliary, nishimura2020discontinuous}, our work uses this idea to provide new insights into the state-of-the-art sticky sampler and to contrast it with novel variants.
We prove that the original sticky sampler is in fact a limit of our Hamiltonian sticky sampler, extending the prior work of \citet{chin2024mcmc}.
Specifically, under limits of increasingly frequent partial refreshments of position and momentum, the Hamiltonian sticky sampler converges strongly to the latent sticky sampler, which in turn converges strongly to the original sticky sampler.

We additionally establish that the latent sampler achieves provably lower asymptotic variance than the original sticky sampler when sampling from product distributions. 
We provide empirical evidence of this efficiency gain holding more generally via a numerical study on high-dimensional linear regression posteriors, where the three sticky samplers can be simulated exactly.
Beyond the linear regression case, our latent and Hamiltonian sticky samplers are more broadly applicable as they can be straightforwardly combined with off-the-shelf, general-purpose numerical integrators.
We demonstrate this generality by applying the numerically-integrated sticky samplers to high-dimensional logistic regression posteriors.


\section{Sticking from a latent universe perspective}

\subsection{Zig-zag samplers as illustrative cases}

The latent universe scheme can be combined with any dynamical Monte Carlo method to yield valid samples.
For illustration, we focus our paper on a specific case of Monte Carlo samplers that depend on \pdmps{} and Hamiltonian dynamics: the \zz{} process sampler \citep{bierkens2019zig} and its Hamiltonian counterpart \citep{nishimura2024zigzag}.

Consider a continuous target $\pi(\pos) \propto e^{-U(\pos)}$ for $\pos \in \mathbb{R}^d$ with differentiable $U$.
The \zz{} sampler augments the target with a velocity variable $\velo$ uniformly distributed on $ \{\pm1\}^d$ and, from an initial condition $(\pos_0, \velo_0)$, follows deterministic dynamics 
\begin{equation}\label{eq:zz_dynamics}
	\pos_t = \pos_0 + t\velo_0, \quad \velo_t = \velo_0
\end{equation}
until the next bounce event.
The event coincides with the earliest of Poisson events that occur according to rates
\begin{equation}\label{eq:zz_rate}
	\lambda_i(\pos_t, \velo_t) = \max\{0, v_{i, t}\partial_i U(\pos_t) \}, \  i=1,\dots, d,
\end{equation} 
where $\partial_i$ denotes the $i$th partial derivative, and results in a sign flip of the corresponding velocity component. 
This \zz{} dynamics has a unique stationary distribution with $\pi(\pos)$ as its marginal and thus can be used to generate samples from the target.

The \hzz{} is a variant of Hamiltonian Monte Carlo that generates Metropolis proposals using Laplace distributed momentum $\mo \in \mathbb{R}^d$ with density $\pi(\mo) \propto \exp(-\sum_i |\mo_i|)$.
The corresponding Hamiltonian dynamics has a velocity $\velo := \diff{\pos}/\!\diff{t} = \mathrm{sign}(\mo) \in \{\pm1\}^d$ and,
like the Markovian dynamics above, its position and velocity $(\pos, \velo = \mathrm{sign}(\mo))$ follow Equation~\eqref{eq:zz_dynamics} in between velocity flip events.
However, the events now occur deterministically when the momentum coordinates, which evolve as
\[
\mo_{i,t} = \mo_{i,0} - \int_0^t v_{i, s}\partial_i U(x_s)\diff{s},
\]
change signs.
The dynamics can be simulated exactly on a piecewise Gaussian target and generates rejection-free proposals \citep{nishimura2024zigzag}.
This deterministic proposal generation is combined with refreshment of the momentum from $\pi(\mo)$ at each iteration, ensuring ergodicity.

In a spike-and-slab model with a likelihood $L(y \mid x)$ and prior $\prod_i \pi_0(x_i)$, our target posterior $\pi(x)$ is the discrete-continuous mixture proportional to $L(y\mid x) \prod_i \pi_0(x_i)$. 
To draw from this class of targets, \citet{bierkens2019zig} introduce the sticky sampler, which combines a standard \pdmp{} with a sticking mechanism at the spike.
Upon reaching 0, the coordinate of the sampler sticks for an amount of time exponentially distributed with mean
\begin{equation}\label{eq:unstick_rate}
	w = \frac{1-\pslab}{\pslab \thinnerspace \slab(0)}
\end{equation} 
while the unstuck coordinates continue to evolve. 
Upon unsticking, the coordinate proceeds with the same velocity as when it reached 0.

\subsection{From point mass to latent universe}

We now present our latent universe scheme that achieves an analogous sticky behavior and, further, allow generalization to other samplers.
As illustrated in Figure~\ref{fig:augment}, the idea is to replace the discrete mass with a continuous density over a latent universe and insert it in between the continuous parts of the target.
We choose the universe's width to be $w$, the mean sticking duration \eqref{eq:unstick_rate}, and the density height to be $\slab(0)$.
This choice of width yields the \textit{latent density} that maintains continuity in its height at the interface between the latent universe and the rest of the space;
other choices of width are possible, but our choice likely optimizes the corresponding latent sampler's efficiency (Appendix~\ref{supp:changing_width}).
Collapsing the latent universe to 0 via a map 
\begin{equation}
	\label{eq:map_from_latent_to_orig}
	\latpos_i \to 0 \ \text{ if } \, | \latpos_i | \leq w/2, \quad
	\latpos_i \to \latpos_i - \operatorname{sign}(\latpos_i) w/2 \ \text{ if } \, | \latpos_i | > w/2
\end{equation}
transforms the latent density back to the original discrete-mixture target.

In other words, our construction yields a continuous density representation $\pi_0(\latpos_i)$ of the spike-and-slab prior in the latent parameter space. 
The corresponding posterior is given as $\pi(\latpos) \propto L(y \mid \latpos) \prod_i \pi_0(\latpos_i)$, where the likelihood is a constant function of $\latpos_i$ and takes values 
$L(y \mid \latpos) 
= L(y \mid \latpos_{-i}, \pos_i = 0)$ 
on $\latpos_i \in [-w/2, w/2]$. 
Figure~\ref{fig:3dpost} illustrates the latent parameter posterior in a two-dimensional case.

To connect the above construction with the sticky behavior, consider running a standard \pdmp{} sampler on the latent density.
When mapping the latent universe to 0 and viewing the trajectory in this original space, the sampler now appears to stick at 0; we refer to this mapped trajectory as the \textit{latent sticky sampler}.
The difference from the original method of \citet{bierkens2023sticky} is that the sticking time is now given deterministically as the time it takes to traverse the universe, which by our construction coincides with the mean of the original method's random sticking time.

By reducing the problem to sampling from a continuous density, our latent universe perspective opens up an array of possibilities to deal with the spike-and-slab posteriors.
In particular, we can apply the \hzz{} on the latent density, yielding \textit{Hamiltonian} sticky samplers.

\begin{figure}
	\centering
	\vspace{2mm}
	\includegraphics[width=13cm]{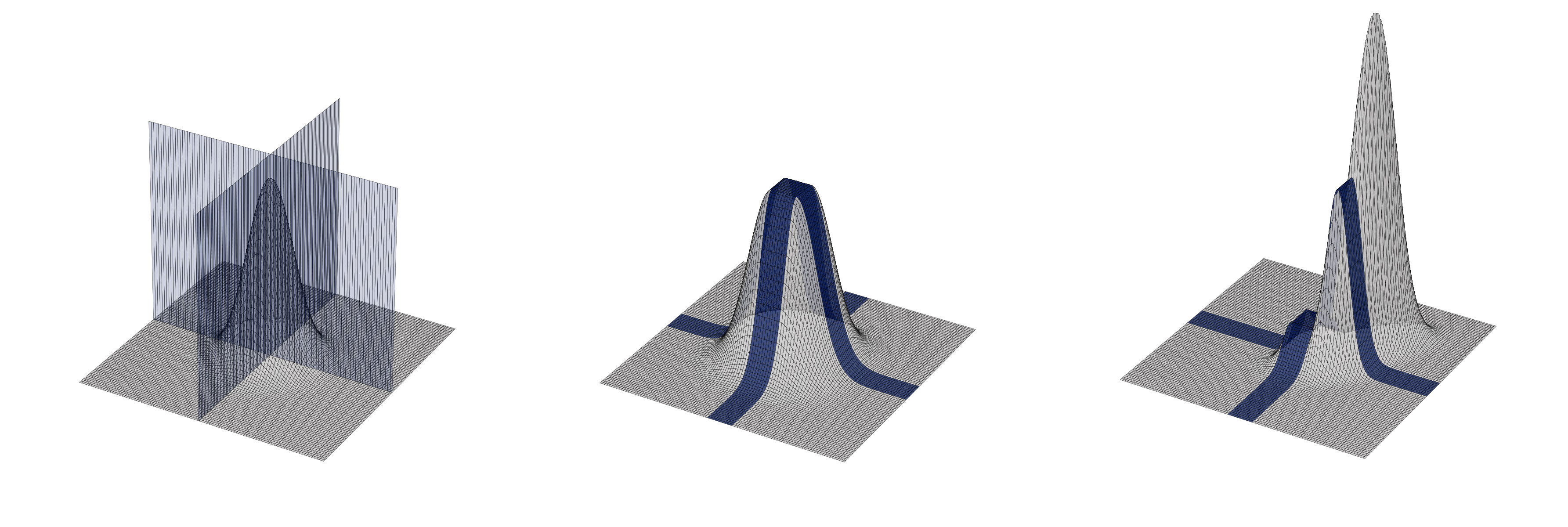}
	\caption{Left: Bivariate product of spike-and-slab priors with standard normal slab, with the spike masses shown in blue. 
		Middle: Latent continuous density representation of the prior. 
		Right: Posterior density in the latent parameter space.%
	}
	\label{fig:3dpost}
\end{figure}


\subsection{Connecting the original, latent, and Hamiltonian sticky samplers}
\label{sec:connecting_samplers}

We now show that the three sticky variants become equivalent under a partial position-momentum refreshment limit.
The latent and Hamiltonian sticky sampler are simply transformations of the corresponding standard samplers on the continuous latent density;
their equivalence in limit is thus implied by the previous results of \cite{chin2024mcmc}, who show that a generalized Hamiltonian dynamics converge strongly to corresponding \pdmps{} under increasingly frequent partial momentum refreshment.

To complete the equivalence among the three, therefore, it only remains to establish a connection between the original and latent sticky samplers.
We achieve this by 
considering the following partial refreshment, which occurs as a Poisson event with rate $r$, of the position component in the latent sticky sampler:
when a position coordinate is inside the latent universe, we resample it uniformly from the universe.
In other words, the latent sticky sampler with partial position refreshment is obtained by considering the \zz{} sampler in the latent space with additional event rate
$r \indicator\{ | \latpos_{i, t} | < w/2 \}$
and corresponding transition kernel 
$\latpos_i \sim \mathrm{Unif}(-w/2, w/2)$
and then mapping it back to the original space via collapsing of the latent universe.
This partial refreshment induces randomness in the latent sampler's sticking times which, in the limit $r \to \infty$, become exponentially distributed with the same mean as the original sticky sampler.
In fact, the partially-refreshed latent sticky sampler converges strongly to the original sticky sampler, as we establish now.

In the theorem statement below, $\rho_T$ denotes the Skorokhod metric on $[0,T]$ \citep{billing}.
The assumed differentiability of the posterior's continuous density part holds whenever both likelihood and prior slab density are differentiable. 
The proof is in Appendix~\ref{supp:pf1}.
\begin{theorem}\label{thm:limit1}
	For a spike-and-slab posterior with twice continuously differentiable density part, the latent sticky sampler with the partial position refreshment converges strongly to the original sticky sampler;
	i.e., from the same initial state, we can create a sequence of the position-velocities trajectories  $(\pos^{\La,r}_t, \velo^{\La,r}_t)$ coupled to the original sticky sampler $(\pos^{\Or}_t, \velo^{\Or}_t)$ so that for any $\epsilon>0$,
	\[
	\lim_{r \to \infty}P[\rho_T\{(\pos^{\La,r}_t, \velo^{\La,r}_t), (\pos^\Or_t, \velo^\Or_t)\} > \epsilon] = 0.
	\]
\end{theorem}
The above result completes the equivalence among the original, latent, and Hamiltonian sticky samplers in the partial refreshment limits, thereby adding to the emerging body of work providing theoretical and practical connections between the \pdmp{} and \hmc{}.
Further, an additional connection to the reversible jump \pdmp{} of \citet{chevallier2023reversible} is explored in Appendix~\ref{supp:changing_width}.

Theorem~\ref{thm:limit1} quantifies the similarity between the original and latent sticky samplers.
On the other hand, the two samplers turn out to have a measurable difference in practical performance.
Specifically, the latent sampler's removal of the randomness in sticking times leads to improved performance, as quantified in the following theorem.
The required regularity condition, as given by Equation~\eqref{eq:regularity_for_regeneration_thm}
in Appendix~\ref{supp:pf2}, is mild and holds in particular when the likelihood is log-concave, the prior slab is Gaussian, and $f$ is polynomially bounded.
\begin{theorem}
	\label{thm:latent_as_rao_blackwellized_ver}
	For a spike-and-slab posterior target with independent coordinates, 
	the estimator $\tau^{-1} \int_0^\tau f(\pos_{i,t})\diff t$ of the marginal statistics $\mathbb{E}[f(\pos_i)]$ achieves a smaller asymptotic variance under the latent sticky sampler than under the original sampler.
\end{theorem}
The main proof idea is that, along each coordinate,  we can decompose the time average into the spike and slab components:
\begin{align*}
	\int_{0}^{\tau}f(x_{i,t})\diff t &= \left| \{ t: x_{i,t} = 0 \} \right|  f(0) + \int_{\{t: x_{i,t} \neq 0\}}f(x_{i,t})\diff t,
\end{align*}
where $| \{ t: x_{i,t} = 0 \} | $ represents the total time spent at 0.
The latent sampler achieves reduction in the variance of the former term, while performing the same as the original sampler on the latter term since the two share the same dynamics away from 0.
The full proof is in Appendix~\ref{supp:pf2}.

While the above result pertains to posteriors with independent coordinates, it is heuristically reasonable to expect such variance reduction when replacing random quantities with their expected values and providing a form of Rao-Blackwellization \citep{robert2021rao_blackwell_mcmc}.
We thus expect our latent universe scheme to provide variance reduction more generally, as we empirically confirm in Section~\ref{sec:simulation}.

\section{Simulation study}\label{sec:simulation}
We now compare the sticky samplers with a simulation study. 
First, we consider a linear regression setting with a Gaussian likelihood, where exact simulations of the dynamics are possible through analytical calculations.
Then, we consider a logistic regression setting, where we apply numerically-integrated dynamics and demonstrate how our latent universe framework extends beyond the Gaussian likelihood case.

\subsection{Data generation and model specification}\label{subsec:data_gen}
Our examples emulate problems in statistical genetics, where predictors exhibit strong block correlations within each chromosome.
Within the $\ell$th block, we simulate the predictors $\{ g_{k \ell} \}$ in a correlated manner via
\begin{align*}
	g_{1 \ell} \sim N(0,1),\quad	g_{k\ell} = \alpha g_{k-1,\ell} + \sqrt{1-\alpha^2}\xi_{k \ell}, \quad \xi_{k \ell} \sim N(0,1),
\end{align*}
while keeping them independent across the blocks. 
We set the value of $\alpha$ as $0.5$, $0.9$, and $0.99$, creating scenarios with varying degrees of correlation.
We simulate $\ell = 1, \ldots, 20$ blocks of $k = 1, \ldots, 100$ predictors, for a total of $p=2{,}000$ predictors. 
We then simulate $n=2{,}000$ observations.

For linear regression, the response is $y = g\pos + \epsilon$ with $\epsilon \sim N(0, \sigma^2 I)$, where 20 of the coefficients $\pos_j$ are randomly chosen to be $\pm 1$ and the rest set to 0.
We set $\sigma = 10$ for the signal-to-noise ratio $\mathrm{var}(g\pos)/\mathrm{var}(y)$ to fall in the range 0.1 to 0.2 across all the $\alpha$ values, reflecting typical strengths of signals in genetics \citep{dun2024robust}.
For logistic regression, we use the same design matrix to retain the same correlation structure among the predictors but simulate binary outcomes via $y \sim \operatorname{Bernoulli}\{\operatorname{logit}^{-1}(x_0 + g \pos)\}$.
We use the same true coefficient values of $\pos$ as in the linear case. 
For the intercept, we set $x_0=-8$ to yield an outcome $y$ with prevalance of around 5\%.

We infer $\pos_j$ based on independent spike-and-slab priors with a standard normal slab.
To assess how the samplers' performances depend on posterior sparsity levels, we experiment with three different values of $\pslab$: $0.001$, $0.01$, and $0.1$, covering under-, exact, and over-estimates of the true sparsity.
For simplicity, we take $\sigma^2$ as known in the linear regression case and focus on the inference on $\pos$.
In the logistic case, we give $x_0$ a flat prior $\pi(x_0) \propto 1$, which can be jointly sampled with $\pos$ by setting its prior slab probability to be 1 and prior slab standard deviation to be infinite.

\subsection{Linear regression sampler setup and evaluation metrics}
\label{sec:sampler_setup_and_metrics}
We compare the original, latent, and Hamiltonian sticky samplers.
The Hamiltonian sticky sampler requires tuning of the travel time; i.e.\ how long to simulate the Hamiltonian dynamics to obtain the next state.
We tune it to achieve good performance on the $\alpha=0.9$, $\pslab=0.01$ case, trying travel times uniformly drawn on $\mathrm{Unif}(0.5 \tau, 1.5 \tau)$ for $\tau=1,2,3,4,5$ and finding $\mathrm{Unif}(2, 6)$ to be optimal.
This same travel time is used for all other cases, though the sampler's performances there can be likely improved with further tuning.
We also slightly modify the momentum refreshment step of the underlying \hzz{}:
instead of the usual full refreshment at each iteration, we only refresh the coordinates $\mo_i$ corresponding to the position coordinates with $| \latpos_i | > w / 2$.
This reduces wasteful backtracking by ensuring that, once entering the latent universe $\latpos_i \in [-w/2, w/2]$, the position coordinate  maintains a persistent motion across iterations and always comes out on the other end.
The original and latent sticky samplers do not require tuning in our simplified setup with fixed $\sigma^2$.

For the original and latent sticky samplers, which are continuous time stochastic processes, we collect posterior samples at  $\tau = 4 m$ for $m \in \mathbb{Z}^+$ so per-iteration travel time is equal to the Hamiltonian sticky sampler's average travel time.
In the experimental settings where memory usage becomes excessive, we thin the samples across all methods equally.

We compare the samplers' performances in terms of effective sample sizes normalized by computation time.
A common way to summarize a sampler performance is to take coordinate-wise posterior means as the statistics of interest, calculate their effective sample sizes, and report the minimum value \citep{hoffman2014no}.
However, we have found these coordinate-wise metrics to be unreliable for the true-zero coefficients because their posterior samples are mostly zeros and show little variation.
For the true-zero coefficients, therefore, we take as alternative statistics the sums of their squared values within each block of the correlated predictors.
We combine these 20 statistics with the coordinate-wise means of the 20 true-nonzero coefficients, calculate their effective sample sizes, and use the minimum of these 40 values as our performance metric.

We calculate each effective sample size measure by averaging the estimates from five independent sampler runs.
Each run is long enough to ensure all the effective sample sizes to be at least 200.
For the latent and Hamiltonian sticky samplers, true-zero coefficients are initialized by a uniform draw within the latent universe, while true-nonzero coefficients were initialized at their true values with small $N(0, 0.001^2)$ perturbations added. 
We run simulations on the  Johns Hopkins Joint High Performance Computing Exchange cluster, allocating a single Intel Xeon Platinum 8558U CPU core and 5 gigabytes of memory for each chain.
Code to reproduce the results is available at 
{\small \url{https://github.com/chinandrew/smoothing\_out\_sticking\_points}}.

\subsection{Linear regression results}
\label{sec:num_results}

\begin{table}[t]
	\caption{Minimum effective sample size per computation time, shown as the ratio relative to the original sticky sampler's performance.}
	\label{tab:results_time}
	\begin{tabular}{c ccc ccc}
		& \multicolumn{3}{c}{Latent sampler} & \multicolumn{3}{c}{Hamiltonian sampler} \\
		& \multicolumn{3}{c}{$\pslab$} & \multicolumn{3}{c}{$\pslab$} \\
		\cmidrule(lr){2-4} \cmidrule(lr){5-7}
		\hspace*{1ex} $\alpha$ \hspace*{1ex} & 0.001 & 0.01 & 0.1 & 0.001 & 0.01 & 0.1 \\
		\cmidrule[0.5pt](lr){1-1} \cmidrule[0.6pt]{2-7}
		0.5  & 1.83 & 1.69 & 1.36 & 1.85 & 1.86 & 2.00 \\
		0.9  & 1.36 & 1.53 & 1.22 & 2.51 & 2.63 & 2.17 \\
		0.99 & 1.09 & 1.14 & 1.26 & 2.50 & 2.60 & 3.69
	\end{tabular}
\end{table}

Table~\ref{tab:results_time} summarizes the latent and Hamiltonian sticky samplers' performances relative to the original sticky sampler's, with values above 1 indicating superior performances.

The latent sampler outperforms the original in all cases.
This empirical finding reinforces our conjecture from Section~\ref{sec:connecting_samplers} that the latent scheme improves on the original sticky paradigm beyond the case theoretically guaranteed by Theorem~\ref{thm:latent_as_rao_blackwellized_ver}.
Note that their relative performance remains essentially identical when standardized by per iteration/integration time, in place of per-computation time;
this is because the original and latent samplers differ only in their sticking times and generating the additional exponential random variables for the original sampler incurs negligible cost.

The Hamiltonian sampler further improves on the latent sampler.
Specifically, its relative advantage over both latent and original samplers generally increases as the predictor correlation, and hence the posterior correlation among parameters, increases with higher $\alpha$ values. 
This finding is consistent with the observation by \citet{nishimura2024zigzag} that the Hamiltonian sampler has an increasing advantage over the \pdmp{} as the parameter correlation increases, though their work does not consider targets with discrete masses.

\subsection{Logistic regression sampler setup and evaluation metrics}\label{subsec:logreg_setup}
The previous example focuses on the case in which the likelihood and slab are both Gaussian and the dynamics underlying the samplers can be simulated exactly.
The latent universe framework is more generally applicable, however, as it can be used with other \hmc{} and \pdmp{} variants, including those based on numerically integrated dynamics.

Using the logistic regression model, we now compare three sticky samplers that take advantage of numerical integration techniques.
The first is the unadjusted zig-zag, given as Algorithm~1 in  \citet{bertazzi2025piecewise}.
We also considered their Metropolis-adjusted Algorithm~3 but found the required momentum flip upon rejection, necessary for the skew detailed balance, to be a significant shortcoming when dealing with latent spike-and-slab posteriors; 
this issue is explained further in Appendix~\ref{supp:adjusted_zz}.
The second is the Hamiltonian zig-zag based on the midpoint integrator (Algorithm S1) of \citet{nishimura2024zigzag}.
Finally, we also run the standard Hamiltonian Monte Carlo with Gaussian momentum and the leapfrog integrator, the algorithm which forms the basis for off-the-shelf software such as Stan \citep{stan} and PyMC \citep{pymc2023}.

For the Hamiltonian samplers, 
we choose a step size which roughly achieves the target acceptance rate of 80\%, as in Stan's default setting, and use the same tuning procedure as before for the integration time.
For the numerically integrated latent zig-zag, we follow Example 6.1 of \citet{bertazzi2025piecewise} and choose a step size $2L^{-1/2}$, where $L$ is the Lipschitz constant of the gradient of the posterior slab's log density.
We find this choice to provide minimal bias in the unadjusted zig-zag, though a larger step size could improve mixing at the expense of increased bias. All step sizes are jittered by 20\% at each iteration.

We otherwise use the exact setup, hardware, and evaluation metrics as the linear regression case.

\subsection{Logistic regression results}
\label{subsec:logreg_results}

\begin{table}[t]
	\caption{Minimum effective sample size per computation time, shown as the ratio relative to the numerically integrated latent zig-zag's performance}
	\label{tab:results_numerical}
	\begin{tabular}{c ccc ccc}
		& \multicolumn{3}{c}{Gaussian momentum} & \multicolumn{3}{c}{Hamiltonian zig-zag} \\
		& \multicolumn{3}{c}{} & \multicolumn{3}{c}{(Laplace momentum)} \\
		& \multicolumn{3}{c}{$\pslab$} & \multicolumn{3}{c}{$\pslab$} \\
		\cmidrule(lr){2-4} \cmidrule(lr){5-7}
		\hspace*{1ex} $\alpha$ \hspace*{1ex} & 0.001 & 0.01 & 0.1 & 0.001 & 0.01 & 0.1 \\
		\cmidrule[0.5pt](lr){1-1} \cmidrule[0.6pt]{2-7}
		0.5  & 0.90 & 0.42 & 0.22 & 1.07 & 0.83 & 0.48 \\
		0.9  & 3.24 & 0.98 & 0.63 & 3.87 & 3.24 & 1.13 \\
		0.99 & 5.94 & 3.41 & 1.69 & 7.37 & 6.63 & 2.85
	\end{tabular}
\end{table}


Results are presented in Table~\ref{tab:results_numerical}.
Mirroring the results from the linear regression case (Table~\ref{tab:results_time}), the Hamiltonian samplers have increasing advantage over the latent sampler as $\alpha$, and hence the posterior correlation, increases.
On the other hand, unlike in the linear regression case where the relative performance is stable over varying $\pslab$, here the Hamiltonian samplers appear to have greater relative advantage over the latent sampler under smaller values of $\pslab$.
This is likely attributable, at least partially, to the difference in their selections of step sizes for numerical integration.
The latent sampler's step size is chosen based on the target's Lipschitz constant, while the Hamiltonian samplers' are tuned to achieve the 80\% acceptance rate,
with the latter approach being potentially more adaptable to the different sparsity levels induced by different $\pslab$ values.

We also observe that the Laplace momentum-based Hamiltonian sampler outperforms the Gaussian momentum-based one in every case.
It is unclear to what extent this pattern holds under other settings, however, particularly because few systematic comparisons are currently available in the literature regarding relative performances of \hmc{} variants based on alternative momenta. 
Our result here suggests a need for more general and systematic performance comparisons under the two momenta, especially given that the Laplace momentum is provably more robust in certain cases \citep{livingstone2019kinetic}.

Finally, we note that the likelihood, and hence the posterior, has discontinuous derivatives at the boundaries of the latent universe.
While existing integrators are generally not designed with discontinuous derivatives in mind and some performance concerns with the leapfrog-integrated sampler have been raised in specific contexts 
\citep{dinh2024hamiltonian}, we find all samplers to be effective for our example.
In particular, the error in Hamiltonian incurred by numerical approximation, i.e.\ the difference in log joint density between the initial and end points of the approximated trajectory, never exceeds 10 across any run within our numerical study.

\section{Discussion}\label{sec:disc}
In this article, we have introduced a new perspective, via the latent universe construction, on the recently developed sticky samplers based on  \pdmps{}.
This perspective has not only provided novel theoretical connections with the Hamiltonian Monte Carlo paradigm, but also motivated more efficient and generalizable methods to carry out posterior inference under spike-and-slab priors via off-the-shelf numerical approximation techniques.

We can also adapt the latent universe method to allow the use of nonlocal-style priors \citep{johnson2012bayesian} as a slab.
More precisely, given a nonlocal slab with $\slab(\pos_i) \geq c > 0$ for  $| \pos_i | > \epsilon$ and $\slab(\pos_i) = 0$ for ``negligible'' values $\pos_i \in [-\epsilon, \epsilon]$ \citep{george1993variable}, we can construct a latent continuous density by inserting the spike mass as a continuous density in between the two positive density parts.
This construction yields a latent sampler whose trajectory mapped to the original space appears to jump from $\pos_i = \pm \epsilon$ to $0$, stick there for the specified amount of time, and exits to the other side with $\pos_i = \mp \epsilon$. 
This possibility shows how our framework further expands the practical scope of spike-and-slab priors for applied Bayesian modeling.

\begin{appendix}

\section{Proof of Theorem \ref{thm:limit1}}\label{supp:pf1}

\begin{proof}[(One-dimensional case)]
	The proof for the one dimensional case contains all the essential ideas, so we begin with this case before moving to higher dimensions.
	As described in the theorem statement, we will construct a coupling of the latent sticky samplers to the original sticky sampler so that they converge in probability with respect to the Skorokhod metric on $[0,\Tt]$ \citep{billing}. 
	For brevity, we drop the $\La$ and $\Or$ superscripts from $(\pos^{\La,r}_t, \velo^{\La,r}_t)$  and  $(\pos^{\Or}_t, \velo^{\Or}_t)$, so the partially refreshed latent samplers are simply denoted as $(\pos^{r}_t, \velo^{r}_t)$  and  the original sampler as $(\pos_t, \velo_t)$.
	
	The Skorokhod metric defines a distance between two real-valued processes $x(t),x'(t): [0,\Tt] \to \mathbb{R}$ while allowing for some dilation in $t$.
	More precisely, it considers a space $\mathcal{K}$ of strictly increasing time dilation functions $\kappa$ with $\kappa(0)=0$ and $\kappa(\Tt) = \Tt$, and quantifies how much $\kappa$ deforms the domain $[0, \Tt]$ by defining a norm on $\mathcal{K}$ as 
	\[
	\lVert \kappa \rVert^\circ = \sup_{t_1\neq t_2} \left | \log \left\{ \frac{\kappa(t_1)-\kappa(t_2)}{t_1-t_2} \right\}\right |.
	\]
	Equipped with these notions, the Skorokhod metric is defined as 
	\begin{equation}
		\label{eq:skorokhod_metric}
		\rho_\Tt(x, x') = \inf_{\kappa \in \mathcal{K}}\{ \max(\lVert \kappa \rVert^\circ, \lVert x - x' \circ \kappa  \rVert_\infty) \},
	\end{equation}
	where $\circ$ denotes a function composition and $\lVert x \rVert_\infty = \sup_t\lVert x(t) \rVert$.
	Under the above distance, the two samplers are considered close as long as they take similar values up to some time dilation, even if their raw values as measured by $\lVert x - x' \rVert_\infty$ are not.

	Away from 0, the two sticky samplers share the same transition kernel. 
	We can therefore couple the samplers to have identical paths except for the duration of their sticking times; 
	details on this part of the coupling are provided later in the proof for the multi-dimensional case. 
	As for the sticking times, Lemma~\ref{lemma:transitionprob} below shows that the latent sampler's times converge in distribution to the original's as $r \to \infty$. 
	We take advantage of this fact and couple the latent sampler's $n$th sticking time $\varsigma_n^r$ to the original's $\varsigma_n$ through the inverse transform method; 
	i.e.\ we generate them by applying the inverses of their cumulative distribution functions to the shared uniform random variable $s_{n} \sim \operatorname{Unif}(0, 1)$.
	With this construction, the distributional convergence implies that $\varsigma_n^r$ converges to $\varsigma_n$ almost surely. 
	The blue and green trajectories of Figure~\ref{fig:vx_kappatilde} illustrate our coupling of the two processes.

	Having constructed the two processes whose sticking durations converge almost surely and who otherwise follow identical paths, we intuitively expect their convergence in the limit.
	We formalize this by showing that the Skorokhod distance between the two processes converges to 0 for almost every realization.
	The idea is to upper-bound the distance by constructing a time deformation function $\kappa_r \in \mathcal{K}$ that aligns the realized sticking times of the latent sampler to the originals, thereby making the distance between $(\pos^r \circ \kappa_r, \velo^r \circ \kappa_r)$ and $(\pos, \velo)$ negligible.
	The almost sure convergence of the sticking times then ensures the norm $\| \kappa_r \|^\circ$ to be small.
	The construction of $\kappa_r$, mathematical details on which we provide below, along with that of $\kappa_r'$ introduced as an intermediate step, is visually illustrated in Figure~\ref{fig:kappa}.
	
	
	To describe our construction of $\kappa_r$, we first introduce some notation. 
	Let $\delta^r_n =  \varsigma_n - \varsigma_n^r$ be the difference between the original and latent samplers' $n$th sticking times.
	Define $\tau_j$ as the inter-sticking durations of the original sampler after the $j-1$ sticking event, with $\tau_1= 0$ if the process starts in the stuck state.
	Further define  $S_n = \sum_{j=0}^{n-1} \varsigma_j + \sum_{j=1}^{n} \tau_j$ as the time at which the $n$th sticking event starts, letting $\varsigma_0 = 0$.

	We now construct $\kappa_r$ to contract or dilate time so that the amount of time spent stuck is the same between $\pos^{r} \circ \kappa_r$ and $\pos$.
	We first set aside the constraint $\kappa_r(\Tt) =  \Tt$ and start by defining  a piecewise-linear continuous $\kappa'_r$ with derivative
	\[
	\frac{\diff}{\diff t}
	\kappa'_r(t) = 
	\begin{cases}
		1 & \pos(t) \neq 0 \\
		\frac{\varsigma_n-\delta^r_n}{\varsigma_n} & x(t) = 0 , \ t \in \left[S_n, \varsigma_n + S_n \right).
	\end{cases}
	\]
	This definition aligns the sticking times and thus ensures $\pos^{r} \circ \kappa'_r = \pos$ and $\velo^{r} \circ \kappa'_r =\velo$.
	An example of $\kappa'_r$ and its effect on the trajectories is shown in Figures~\ref{fig:vx_kappatilde} and \ref{fig:kappa}.
	
	\begin{figure}[htb]
		\centering
		\vspace{2mm}
		\includegraphics[width=14cm]{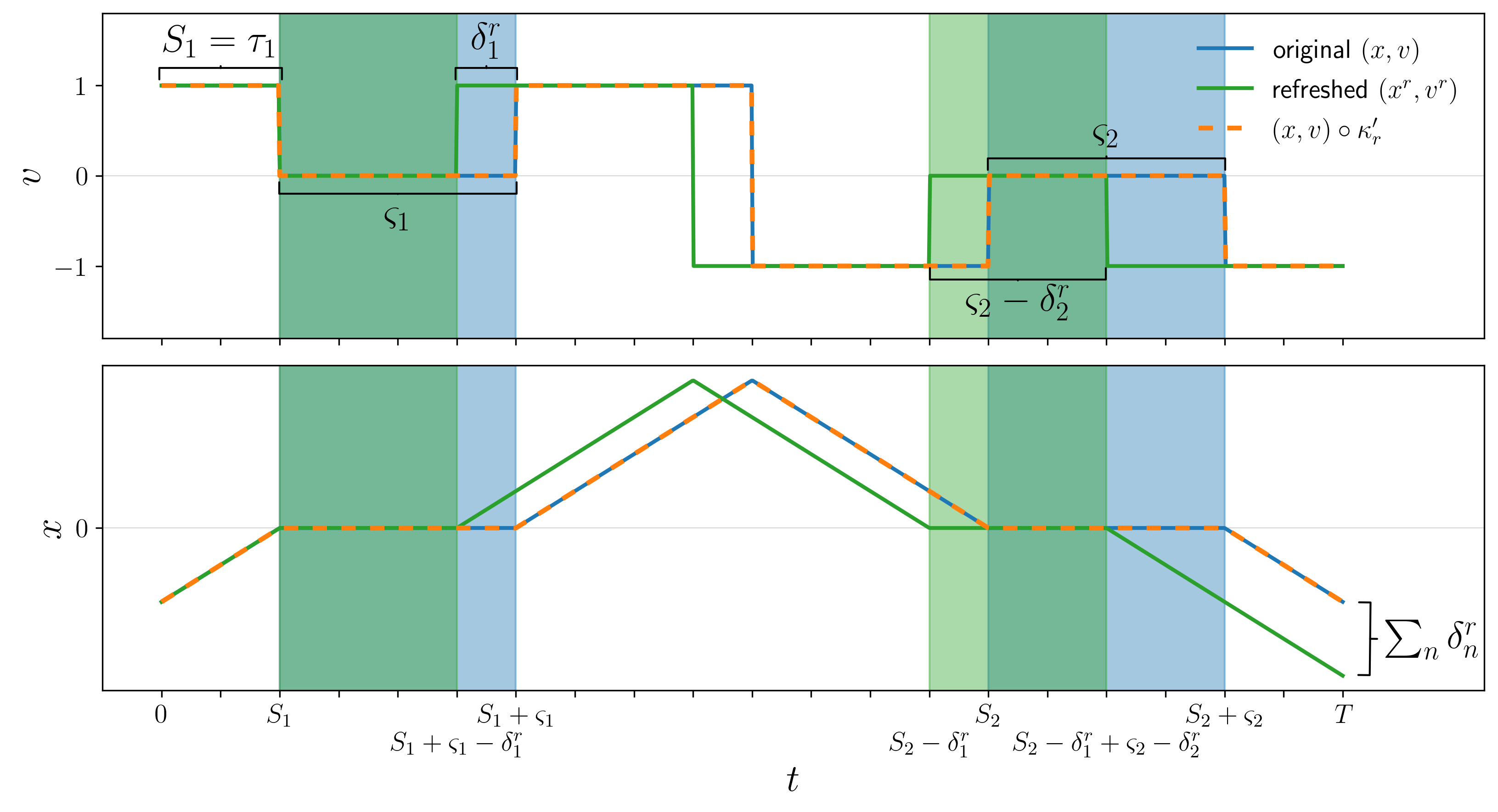}
		\caption{%
			Example process starting unstuck with $x < 0$ and $\velo=1$, getting stuck at the spike, unsticking from there, bouncing once against the gradient, sticking again, and unsticking one last time. In this example, the refreshed process unsticks before the original zig-zag for both sticks. The effect of $\kappa'_r$ applied to $(\pos^r, \velo^r)$ is shown in orange, where $\pos^r(\Tt) \neq \pos^r \circ \kappa'_r(\Tt)$ since $\kappa'_r(\Tt) \neq \Tt$. 
			The blue and green shaded regions represent time spent stuck for the original and refreshed processes, respectively; 
			see Figure~\ref{fig:kappa} for a visual illustration of how $\kappa'_r$ aligns the stuck times of the two samplers.
		}
		\label{fig:vx_kappatilde}
	\end{figure}
	
	\begin{figure}[htb]
		\centering
		\vspace{2mm}
		\includegraphics[width=10cm]{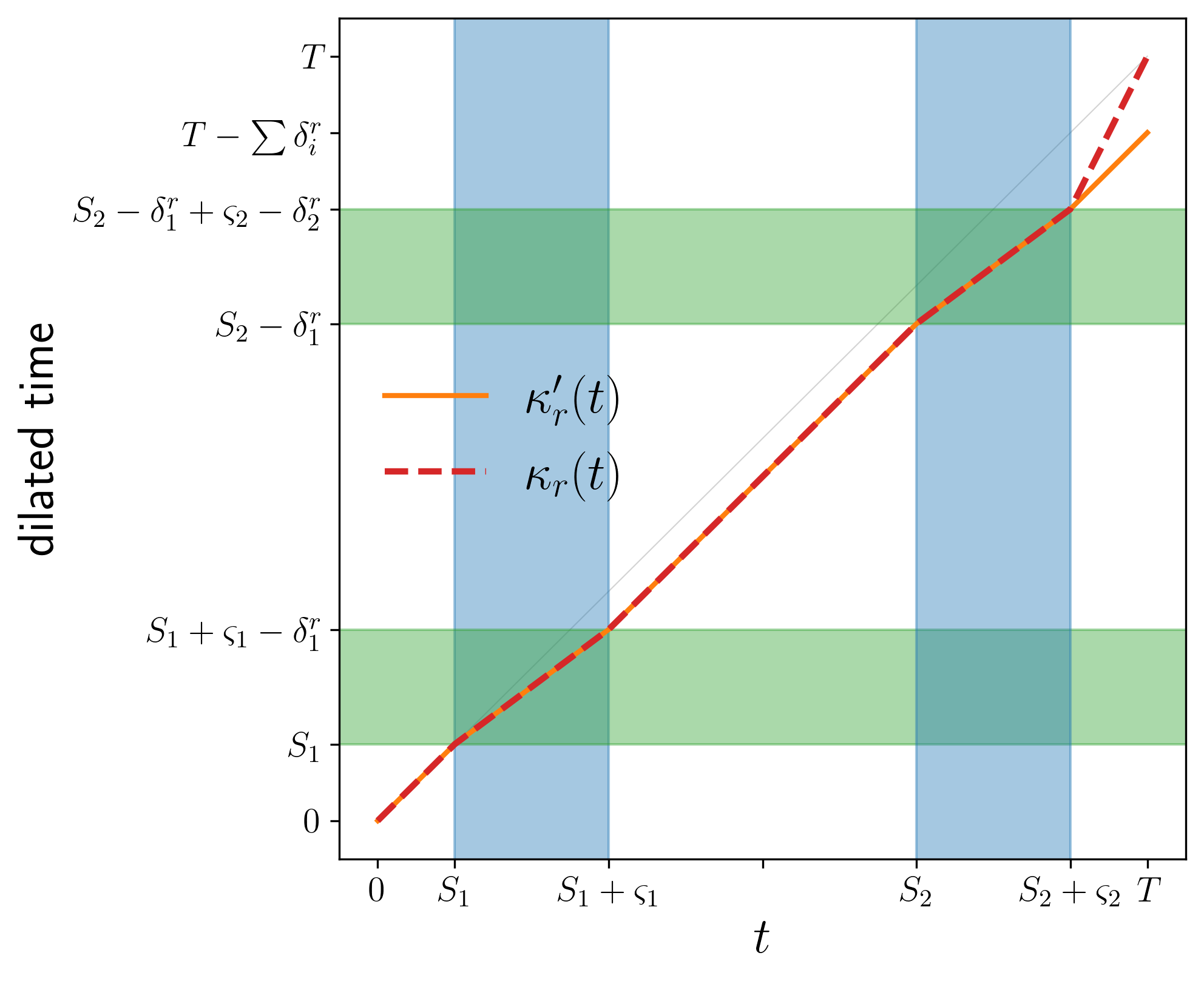}
		\caption{%
			Time dilation functions $\kappa'_r$ and $\kappa_r$ for an interval $[0,\Tt]$. The dilation to align sticking times occurs in the regions where the blue and green bands overlap. Outside of these regions there is no dilation until the final segment between $S_2+\varsigma_2$ and $\Tt$, at which point $\kappa_r$ ``catches up" so as to satisfy $\kappa_r(\Tt) = \Tt$ as required of a dilation function by the Skorokhod metric.
		}
		\label{fig:kappa}
	\end{figure}

	\begin{figure}[htb]
		\centering
		\vspace{2mm}
		\includegraphics[width=14cm]{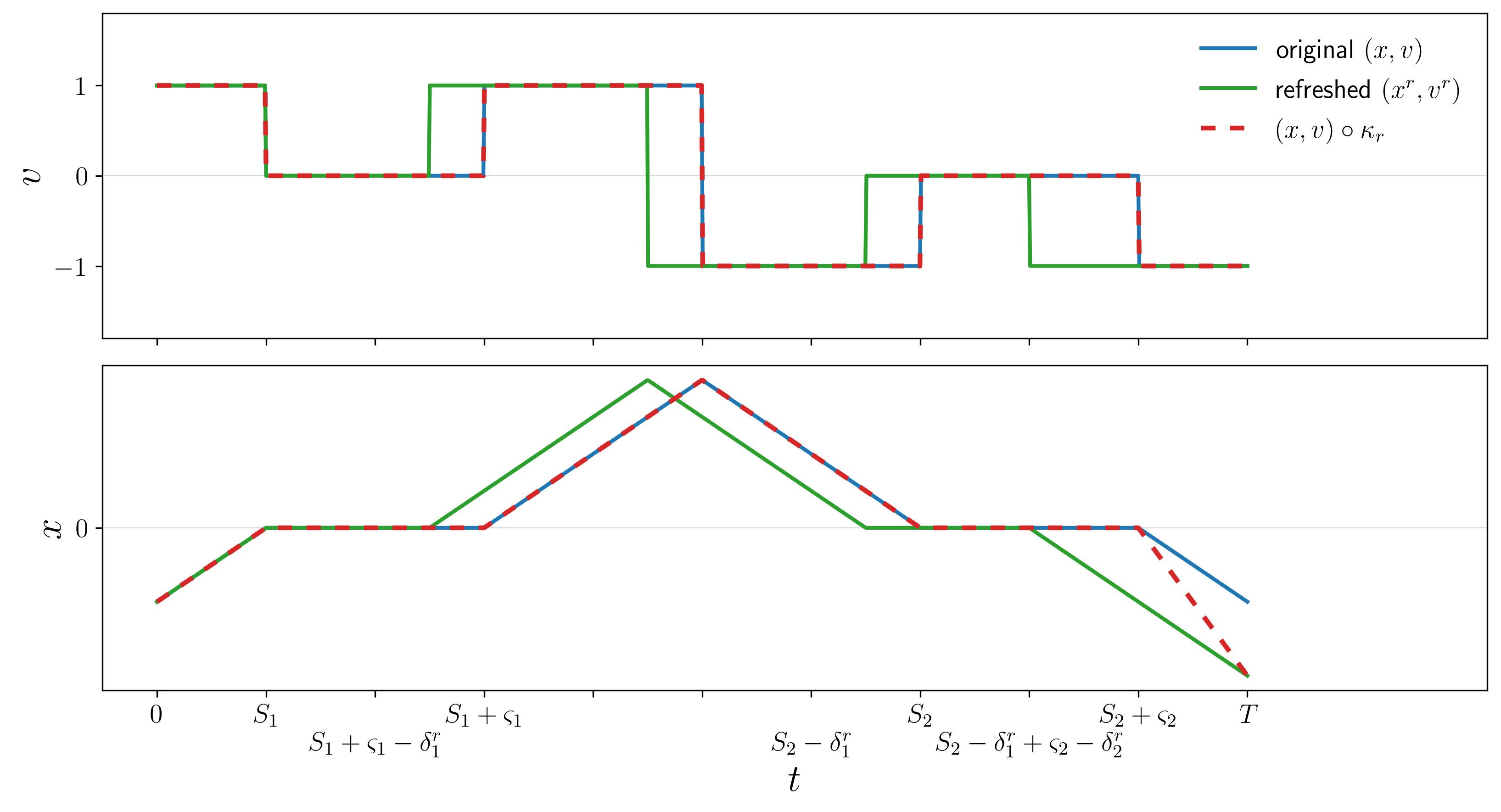}
		\caption{The process from Figure~\ref{fig:vx_kappatilde} with $\kappa_r$ applied to $(\pos^r, \velo^r)$. The velocities are still identical, but the positions are only identical up until the last unsticking event, at which point they diverge up to a maximum discrepancy of $\sum_n\delta_n$.}
		\label{fig:vx_kappa}
	\end{figure}
	
	We now modify the final segment of $\kappa'_r$ to satisfy the constraint $\kappa(\Tt)=\Tt$ as required by the Skorokhod metric (Equation~\ref{eq:skorokhod_metric}).
	Precisely, we modify the final segment to have slope $\frac{\tau_N + \sum_n \delta^r_n}{\tau_N}$ or $\frac{\varsigma_N + \sum_n \delta^r_n}{\varsigma_N}$, depending on if the process is unstuck or stuck at $\Tt$, to ``catch up" for the remainder of the process.
	This yields $\kappa_r$ that satisfies $\kappa_r(\Tt)=\Tt$ and still keep the process $(\pos^r \circ \kappa_r, \velo^r \circ \kappa_r)$ close to $(\pos, \velo)$.
	For the velocity, we have $\velo^{r} \circ \kappa_r =\velo$.
	For the position, we have $\pos^{r} \circ \kappa_r = \pos$ if the processes end in the stuck state.
	If they end in the unstuck state, the processes stay together until the last unsticking event, at which point they deviate linearly up to the maximum $\sum_n \delta^r_n$ discrepancy, i.e.
	\[
	\lVert \pos^{r} \circ \kappa - \pos \rVert_\infty \leq \sum_n \delta_n.
	\]
	The modification of $\kappa'_r$ to $\kappa_r$ and the resulting $(\pos^r \circ \kappa_r, \velo^r \circ \kappa_r)$ are visually illustrated in Figures~\ref{fig:kappa} and \ref{fig:vx_kappa}.
	
	We have thus constructed a time dilation $\kappa_r \in \mathcal{K}$ such that
	$\| (\pos^r \circ \kappa_r, \velo^r \circ \kappa_r) - (\pos^r, \velo^r) \|_\infty \leq \sum_n \delta_n^r$ and $\lVert \kappa_r \rVert^\circ \leq \max \left\{ \frac{\tau_N + \sum_n \delta^r_n}{\tau_N}, \, \frac{\varsigma_N + \sum_n \delta^r_n}{\varsigma_N} \right\}$.
	The Skorokhod distance $\rho_\Tt \left(  (\pos^r \circ \kappa_r, \velo^r \circ \kappa_r), (\pos^r, \velo^r) \right)$ is bounded by the minimum of the two right hand sides, both of which converge to 0 almost surely by Lemma~\ref{lemma:transitionprob}.
\end{proof}

\begin{lemma}\label{lemma:transitionprob}
	For the latent sticky sampler with position refreshment $(\pos^r_t, \velo^r_t)$ as defined in Theorem \ref{thm:limit1}, its sticking times converge in distribution to an exponential random variable with mean $w$ as $r \to \infty$.
\end{lemma}

\begin{proof}
	We focus on a single sticking event since all the sticking times are independent and identically distributed.
	Assume for simplicity that the latent sticky sampler has entered the latent universe at 0 from the left with velocity $+1$.
	Let $\varsigma^r$ denote the time that the sampler leaves the universe.
	Our goal is to show for all $s$ that the probability of the event $\{ \varsigma^r > s \}$, which indicates the sampler remaining in the universe for duration $[0, s]$, satisfies
	\[
	\lim_{r \to \infty} P( \varsigma^r > s) = \exp\left(-\frac{s}{w}\right).\]
	Define $s^r_1, s^r_2, \dots$ as cumulative sums of independent $\operatorname{Exp}(r)$ random variables, representing the sequence of potential refresh times, and define $M_\mathrm{ref}^r = \max_k \{ s_k^r < s \}$, representing the number of potential refresh events before time $s$ which would occur provided the sampler does not leave the latent universe before time $s$.
	Instead of analyzing $P( \varsigma^r > s)$ directly, we analyze the conditional probability
	\begin{equation}\label{eq:unstick_cond_prob}
		P_{\mathrm{cond}, s}(m) = P( \varsigma^r > s \mid M_\mathrm{ref}^r = m).
	\end{equation}
	We denote the conditional probability as $P_{\mathrm{cond}, s}(m)$ without indication of $r$ since, once we condition on $M_\mathrm{ref}^r$, the distribution of $ \varsigma^r$ does not depend on $r$; our derivations below confirm this.	
	Since 
	$P( \varsigma^r > s ) 
	= \mathbb{E} \left[
	P_\mathrm{cond, s}(M_\mathrm{ref}^r)
	\right]$
	and $M_\mathrm{ref}^r \to \infty$ in probability as $r \to \infty$, we have
	$\lim_{r \to \infty} P( \varsigma^r > s ) = \lim_{m \to \infty} P_{\mathrm{cond}, s}(m)$.
	The proof is complete, therefore, once we establish $P_{\mathrm{cond}, s}(m) \to \exp(-s/w)$ as $m \to \infty$.

	By a property of the Poisson process, the $m$ refreshment times are conditionally distributed as uniforms on $[0,s]$; in turn, the $m+1$ inter event times $q_0s, \dots q_ms$ are distributed such that  $q^m = (q_0, \dots,  q_m) \sim \mathrm{Dirichlet}(1,\dots, 1)$.
	The process remains in the universe until the first refresh if $sq_0 < w$.
	For subsequent refreshes with $k \geq 1$, the process remains if the refreshed position is more than $s q_k$ away the universe's upper boundary, which occurs with probability $(1- sq_k/w)\indicator(sq_k < w)$.
	This fact, combined with the independence across the refreshed positions, allows us to express the conditional probability as
	\begin{align*}
		P_{\mathrm{cond}, s}(m)
		&=\mathbb{E}_{q^m} \left\{
		\mathbb{P} \left(\{sq_0 < w\} \text{ and }
		\bigcap_{k = 1}^m \left\{ \thinnerspace \text{$k$th refreshed position $>s q_k$ away} \right\} \, \middle| \, q^m 
		\right)
		\right\}\\
		&=\mathbb{E}_{q^m} \left\{
		\indicator(sq_0 < w) \prod_{k = 1}^m \indicator(sq_k < w) \left( 1- \frac{sq_k}{w} \right)
		\right\}\\
		&=\mathbb{E}_{q^m} \left\{
		\prod_{k = 0}^m \indicator (sq_k < w) 
		\prod_{k = 1}^m \left( 1- \frac{sq_k}{w} \right) 
		\right\}.
	\end{align*}
	
	We now observe that
	\begin{equation}\label{eq:limit_equivalence}
		\begin{split}
		\lim_{m \to \infty} P_{\mathrm{cond}, s}(m)
		&= \lim_{m \to \infty} \mathbb{E}_{q^m} \left\{
		\prod_{k = 0}^m \indicator (sq_k < w) 
		\prod_{k = 1}^m \left( 1 - \frac{sq_k}{w} \right)  
		\right\} \\
		&= \lim_{m \to \infty} \mathbb{E}_{q^m} \left\{
		\prod_{k = 1}^m \left( 1- \frac{sq_k}{w} \right)
		\right\},
		\end{split}
	\end{equation}
	where the latter equality can be seen to hold as follows.
	The equality is immediate, and in fact holds without taking the limit, when $s \leq w$ and hence $\indicator( s q_k  < w ) = 1$. 
	Focusing on the case $s > w$, therefore, we note that the difference between the two expectations can be expressed as
	\begin{align*}
		0 
		&\leq \mathbb{E}_{q^m} \left\{
		\prod_{k = 1}^m \left( 1- \frac{sq_k}{w} \right)
		\right\} - 
		\mathbb{E}_{q^m} \left\{
		\prod_{k = 0}^m \indicator(sq_k < w) 
		\prod_{k = 1}^m \left( 1 - \frac{sq_k}{w} \right)  
		\right\} \\
		&= 
		\mathbb{E}_{q^m} \left\{
		\left( 1 - \prod_{k = 0}^m \indicator(sq_k < w) \right)
		\prod_{k = 1}^m \left( 1 - \frac{sq_k}{w} \right)  
		\right\} \\
		&= 
		\mathbb{E}_{q^m} \left\{
		\indicator \left( \bigcup_{k = 0}^m \{ sq_k \geq w \} \right)
		\prod_{k = 1}^m \left( 1 - \frac{sq_k}{w} \right)  
		\right\}.
	\end{align*}
	We show that the last expression, and hence the difference, tends to 0 as $m \to \infty$ by establishing an appropriate upper bound.
	To this end, we note that
	\begin{align*}
		\left| \prod_{k = 1}^m \left( 1 - \frac{sq_k}{w} \right)	 \right|  
		&\leq  \prod_{k = 1}^m  \left( 1 + \frac{sq_k}{w} \right) \\
		&\leq  \prod_{k = 1}^m  \exp\left(\frac{sq_k}{w}\right) \\
		&= \exp\left(\frac{s}{w}\sum_{k = 1}^mq_k\right) \\
		&\leq \exp\left(\frac{s}{w}\right), 
	\end{align*}
	and hence that
	\begin{align*}
		\mathbb{E}_{q^m} \left\{
		\indicator \left( \bigcup_{k = 0}^m \{ sq_k \geq w \} \right)
		\prod_{k = 1}^m \left( 1 - \frac{sq_k}{w} \right)  
		\right\}
		\leq 
		P \left( \bigcup_{k = 0}^m \{ sq_k \geq w \} \right)
		\exp\left(\frac{s}{w}\right).
	\end{align*}
	Finally, we use a union bound for the probability and use the fact that each Dirichlet variable $q_k$ is marginally distributed as $\mathrm{Beta}(1, m)$:
	\begin{align*}
		P \left( \bigcup_{k = 0}^m \{ sq_k \geq w \} \right)
		&\leq \sum_{k=0}^m P\left(q_k \geq \frac{w}{s} \right)\\
		&= (m+1)\int_{w / s}^\infty m(1-q)^{m-1} \diff q \\ 
		&= (m+1) \left(1 - \frac{w}{s} \right)^m,
	\end{align*}	
	where the last term tends to 0 as $m \to \infty$ in the case $s > w$ under consideration.
	
	With Equation~\eqref{eq:limit_equivalence} established, we now compute its rightmost expectation:
	\begin{equation}
		\begin{aligned}[b]
			&\hspace*{-1.5em}\mathbb{E}_{q^m} \left\{ \prod_{k=1}^{m} \left(1-\frac{s}{w}q_k\right)  \right\} \\
			&=\mathbb{E}_{q^m}\left\{1 + \sum_k \left( - \frac{s}{w} \right) q_k  + \sum_{k, \ell} \left( - \frac{s}{w} \right)^2 q_k q_\ell + \dots +  \left(-\frac{s}{w}\right)^{m}(q_1\dots q_m) \right\} \\
			&=1+\sum_{k=1}^{m} {m \choose k} \left(-\frac{s}{w}\right)^k \mathbb{E}_q(q_1 \dots q_{k})   \\
			&=\sum_{k=0}^{m} {m \choose k} \left(-\frac{s}{w}\right)^k \frac{m!}{(m+k)!}, \quad \\ 
			&= \sum_{k=0}^{m} \frac{m!m!}{(m-k)!(m+k)!}  \frac{(-s)^k}{w^k k!} \\
			&= \sum_{k=0}^{\infty} \indicator(k \leq m) \left( \prod_{\ell = 0}^{k-1} \frac{m - \ell}{m+k-\ell} \right) \frac{(-s)^k}{w^k k!},
		\end{aligned}\label{eq:probleave}
	\end{equation}
	where the second equality follows from equivalence of expectations due to exchangeability among the $q_k$.
	Taking the limit as $m \to \infty$ of the last expression above yields the desired result:
	\begin{align*}
		\lim_{m \to \infty} \sum_{k=0}^{\infty} \indicator(k \leq m) \left( \prod_{\ell = 0}^{k-1} \frac{m - \ell}{m+k-\ell} \right) \frac{(-s)^k}{w^k k!}
		&=  \sum_{k=0}^{\infty} \lim_{m \to \infty}  \indicator(k \leq m) \left( \prod_{\ell = 0}^{k-1} \frac{m - \ell}{m+k-\ell} \right) \frac{(-s)^k}{w^k k!} \\
		&=  \sum_{k=0}^{\infty}  \frac{(-s)^k}{w^k k!}\\
		&= \exp\left(-\frac{s}{w}\right),
	\end{align*}
	where the first equality follows by the dominated convergence theorem, as the $k$th term in the sum is dominated in absolute value by $s^k/(w^k k!)$.
\end{proof}

\begin{proof}[(Multidimensional case)]
	We now establish the convergence result in dimension $d>1$. 
	
	We start with a description of our coupling between the two sticky samplers in $\mathbb{R}^d$.
	The coupling is constructed using a collection of independent $\operatorname{Unif}(0, 1)$ random variables $\{s_{n,i}, u_{n,i}\}$ for $i=1, \dots, d$ and $n \geq 0$, where $i$ indexes the dimensions and $n$ indexes the sequence of events in between which the velocity stays constant.
	Note that an unstick event does not necessarily follow a stick event in the multidimensional case---a bounce event along another coordinate may occur in between---and we do not a priori know the type of the $n$th event.
	
	We present in detail the construction of the latent sticky sampler $(\pos^r_t, \velo^r_t )$ with partial position refreshment;
	the construction of the original sticky sampler $(\pos_t, \velo_t )$ is essentially identical and is obtained by simply replacing $(\pos^r, \velo^r )$ with $(\pos, \velo )$ in each step below.
	Denote by $(\pos_n^r, \velo_n^r)$ the sampler state immediate after the $n$th event and the corresponding update in the velocity. 
	Since the sampler trajectory is deterministic in between events, the coupling construction is complete once we provide a recipe for determining the time and type of the $n$th event from $(\pos_{n-1}^r, \velo_{n-1}^r)$.
	To this end, we start by defining \textit{potential event times} $\tau^r_{n, i, e}$ for each event type $e \in \mathcal{E} = \{ \mathrm{stick}, \mathrm{unstick}, \mathrm{bounce} \}$.
	For each event type $e$ and coordinate $i$, the potential event time $\tau^r_{n, i, e}$ represent the time at which the event would occur if no other events occur.
	The earliest of them constitutes the actual event time $\tau^r_n = \min_{i, e} \tau^r_{n, i, e}$ with the corresponding event type and index given by the pair $g^r_n = \arg\min_{i, e} \tau^r_{n, i, e}$.
	For the original sticky sampler $(\pos_t, \velo_t )$, we analogously denote its potential and actual event times by $\tau_{n, i, e}$ and $\tau_{n}$ and corresponding event type-index pair by $g_n$.
	
	We construct $\tau^r_{n, i, e}$ from the given random variables $\{s_{n,i}, u_{n,i}\}$ as follows.
	We start with the case $e = \mathrm{unstick}$, where the coupling idea is the same as in the one-dimensional case. 
	For each stuck coordinate with $\pos^r_{n,i} = 0$, we define $\tau^r_{n, i, \mathrm{unstick}}$ as the inverse cumulative distribution transform of $s_{n,i}$;
	Lemma~\ref{lemma:transitionprob} then guarantees that $\tau^r_{n, i, \mathrm{unstick}} \to - \log s_{n,i}$ almost surely as $r \to \infty$.
	For each of the coordinates that are already unstuck, we set $\tau^r_{n, i, \mathrm{unstick}} = \infty$ as a placeholder indicating the impossibility of the event type.
	Turning to the case $e = \mathrm{bounce}$, we construct the potential event time as
	\begin{equation}\label{eq:lszz_bounce}
		\tau^r_{n, i, \mathrm{bounce}} = 
		\inf_{t>0}\left[-\log u_{n,i}  \leq \int_0^t \{\velo^r_{n-1,i} \partial_i U(\pos^r_{n-1} + s \velo^r_{n-1})\}^+\diff{s} \right],
	\end{equation}
	where $\velo^r_{n-1, i}$ refers to the $i$th component of the $\velo^r_{n-1}$  and $\partial_i$ is the $i$th partial derivative. 
	Finally, for $e = \mathrm{stick}$, there is no explicit coupling as the time to the next stick is simply the time until reaching 0:
	\begin{equation}
		\label{eq:lszz_stick}
		\tau^r_{n, i, \mathrm{stick}} = 
		\begin{cases}
			-\pos^r_{n-1, i} / \velo^r_{n-1, i} & \text{if } \, \mathord{-}\pos^r_{n-1, i} / \velo^r_{n-1, i} > 0 \\
			\infty & \textrm{otherwise}.
		\end{cases}
	\end{equation}
	
	The reminder of proof follows the same main idea as in the one-dimensional case, based on the construction of a dilation function to aligns the event times between the two samplers.
	Unlike the one-dimensional case, where both samplers follow the same path away from 0, the samplers can now take different paths if their sticking times differ;
	this is because their unstuck coordinates continue to evolve even when their stuck coordinates both remain at 0.
	However, we will show that the two samplers will still converge as their sticking times converge as $r \to \infty$. 
	
	To establish the convergence of the two samplers, it suffices to establish the convergence of their event times and types.
	If no coordinate is stuck initially, then our coupling ensures that the samplers follow an identical path from the shared initial state.
	Without loss of generality, therefore, we assume that the samplers' initial state has at least one coordinate stuck.
	We start by establishing that the following convergence of the potential event times holds for $n = 1$, corresponding to the first event: $\tau^r_{n, i, e} \to \tau_{n, i, e}$ for all $i$ and $e$ as $r \to \infty$
	This convergence in turn implies that of the actual event type, time, and location: $\tau^r_{n} \to \tau_{n}$, $g^r_n \to g_n$, and $(\pos_n^r, \velo_n^r) \to (\pos_n, \velo_n)$.
	We then show the same holds for $n \geq 2$, i.e.\ for the second and all remaining events, of which there are almost surely finitely many. 
	
	For $n=1$, the shared initial condition and our coupling ensures that the potential sticking and bounce event times, constructed as in Equation~\eqref{eq:lszz_bounce} and \eqref{eq:lszz_stick}, coincide between the two samplers;
	i.e.\ $\tau^r_{1, i, e} = \tau_{1, i, e}$ for $e \in \{ \mathrm{stick},  \mathrm{bounce} \}$.
	And our coupling construction and Lemma~\ref{lemma:transitionprob} ensure the convergence of the potential unsticking event times; i.e.\ $\tau^r_{1, i, \mathrm{unstick}} \to - \log s_{1,i} = \tau_{1, i, \mathrm{unstick}}$ as $r \to \infty$.
	Since all the potential event times converge, so do the actual event times and types of the two samplers.
	
	We now consider the potential event times for the second event, the case $n = 2$.
	The convergence of the two samplers' first event guarantees that, for all sufficiently large $r$, the samplers have the same velocity $\velo^r_1 = \velo_1$ after the first event.
	They have different positions $\pos^r_1$ and $\pos_1$, but the discrepancy $\delta_1(r) = \lVert \pos^r_1 - \pos_1 \rVert$ converges to 0 as $r \to \infty$ and hence can be made arbitrarily small.
	For the potential sticking times \eqref{eq:lszz_stick}, their continuous dependence on $\pos^r_1$ and $\pos_1$ implies that $\tau^r_{2, i, \mathrm{stick}} \to \tau_{2, i, \mathrm{stick}}$ as $r \to \infty$ and $\delta_1(r) \to 0$.
	For the potential bounce times, we note their continuous dependence on the integrals on the right hand side of Equation~\eqref{eq:lszz_bounce}.
	And the integrals converge as $r \to \infty$ because the integrands converge, which can be seen from the following bound on their difference:
	\begin{align*}
		&\hspace*{-1.75em}\left| \{\velo_{1,i} \partial_i U(\pos_1  + s \velo_1)\}^+  -  \{\velo^r_{1,i} \partial_i U(\pos^r_1 + s \velo^r_1)\}^+ \right| \\
		&\leq \left| \velo_{1,i} \partial_i U(\pos_1  + s \velo_1)  -  \velo^r_{1,i} \partial_i U(\pos^r_1 + s \velo^r_1) \right| \\
		&= \left|  \partial_i U(\pos_1  + s \velo_1) - \partial_i U(\pos^r_1 + s \velo_1) \right| \\
		&\leq \left\{ \textstyle \sup_{s \in [0, \Tt]} \lVert \nabla \partial_iU(\pos_1  + s \velo_1) \rVert \right\} \lVert \pos^r_1 - \pos_1 \rVert \\
		&\leq \left\{ \textstyle \sup_{x \in B(\pos_0, \Tt)} \lVert \nabla \partial_iU(x) \rVert \right\} \delta_1(r)
	\end{align*}
	where $B(\pos_0, \Tt)$ denotes the $L_\infty$ ball of radius $\Tt$ centered at $\pos_0$.
	This establishes the convergence of the integrals and hence $\tau^r_{2, i, \mathrm{bounce}} \to \tau_{2, i, \mathrm{bounce}}$.
	Finally, for the potential unsticking times, their convergence $\tau^r_{2, i, \mathrm{unstick}} \to \tau_{2, i, \mathrm{unstick}}$ again follows from our coupling construction and Lemma~\ref{lemma:transitionprob}.
	
	We can analogously establish the convergence $\tau^r_{n, i, e} \to \tau_{n, i, e}$ for the subsequent events and hence $(\pos_n^r, \velo_n^r) \to (\pos_n, \velo_n)$ for all $n \geq 1$.
	This in turn implies $\lVert (x, v) - (x^r, v^r) \circ \kappa^r \rVert_\infty \to 0$ and  $\lVert \kappa^r \rVert^\circ \to 0$, establishing the convergence of the two samplers in the Skorokhod metric.
\end{proof}

\section{Proof of Theorem \ref{thm:latent_as_rao_blackwellized_ver}}\label{supp:pf2}
Our proof utilizes the facts that, when the target has independent coordinates, the coordinate processes of both latent and original sticky zigzags are independent and each coordinate process is \textit{regenerative}; 
i.e.\ there is a sequence of \textit{regeneration times} $T_0, T_1,  \dots$ such that the intervals $T_{j+1}-T_{j}$ and integrals $\int_{T_j}^{T_{j+1}} f(x_{i, t})\diff{t}$ are independently and identically distributed \citep{serfozo2009basics}.
Such regeneration times can be constructed as the return times to pre-specified position $\pos_{i, 0} \neq 0$ and velocity $\velo_{i, 0}$, defining $T_{-1} = 0$ and 
$T_{j} = \inf \left\{ 
t > T_{j - 1} : \pos_i(t) = \pos_{i, 0} 
\, \text{ and } \,
\velo_i(t) = \velo_{i, 0} 
\right\}$
for $j \geq 0$.

The regenerative property of the coordinate process gives the following result, per Theorem 65 of \citet{serfozo2009basics}.
To simplify the notation, we take the initial sampler position to coincide with the regeneration point so that $T_0 = 0$.
Now assume the regularity conditions as below:

\begin{equation}
	\label{eq:regularity_for_regeneration_thm}
	\begin{gathered}
		\tau := \E[T_1] < \infty, \\
		\mu := \frac{1}{\tau}\E\left\{\int_0^{T_1} f( x_{i, t} ) \diff t\right\} < \infty, \\
		\sup_{0\leq t\leq T_1} \left|\int_{0}^{t} f( x_{i, t} ) \diff t - \int_{0}^{T_0} f( x_{i, t} ) \diff t \right|<\infty,  \\
		0<\var\left\{\int_0^{T_1} f( x_{i, t} ) \diff t-\mu T_1\right\}<\infty.
	\end{gathered}
\end{equation}
Then the asymptotic variance of $\{\int_0^\tau f( x_{i, t} ) \diff t -\mu \tau\}/ \tau^{1/2}$ as $\tau \to \infty$ is given by
\begin{equation}\label{eq:asymp_var}
	\frac{1}{\tau}\var\left\{\int_0^{T_1} f( x_{i, t} ) \diff t-\mu T_1\right\}.
\end{equation}
The positive recurrence condition $\E[T_1] < \infty$ holds for the zigzags on a spike-and-slab posterior whenever, for example, the likelihood is strongly log-concave and has a bounded gradient \citep{bierkens2019zig_ergodicity}.
For such a posterior, the rest of the regularity conditions are trivially satisfied when, for example, $f$ is polynomially bounded.

\begin{proof}
	Independence of the coordinate processes means it suffices to consider a single dimension.
	Correspondingly, we simplify the notation by dropping the subscript from $\pos_i$ and denoting it simply as $\pos$.
	By discarding as necessary the processes' initial segments, which have vanishing contributions in the limit, we can without loss of generality assume that both samplers starts away from 0.
	We take this initial state as the point at which the regeneration occurs.
	For both latent and original samplers, the asymptotic variance is then given by the expression~\eqref{eq:asymp_var}, which we decompose into the terms corresponding to the stuck and unstuck states:
	\begin{align*}
		\frac{1}{\tau}\var\left\{\int_{0}^{T_1}f(x_t)\diff t - \mu T_1\right\} &=\frac{1}{\tau}\var\left[\int_{\{t: x_t \neq 0\}}f(x_t)\diff t - \mu (T_1 - S_1) + S_1\left\{f(0) - \mu\right\}\right],
	\end{align*}
	where $S_1 = \left| \{t \in [0, T_1]: x_t = 0\} \right|$ denotes the time spent stuck at 0.
	For the rest of the proof, we will denote the stuck time of the latent sampler as $S_1^\La$ and of the original sampler as $S_1^\Or$;
	we will denote both samplers' trajectories away from 0 as $x_t$ omitting the superscript, with a slight abuse of notation, since the two shares the same transition kernel away from 0.

	The stuck time $S_1^\La$ of the latent sampler is deterministic, so its asymptotic variance is given by
	\[
	\nu^\La := 
	\frac{1}{\tau}\var\left\{\int_{\{t: x_t \neq 0\}}f(x_t)\diff t - \mu (T_1^\La - S_1^\La) \right\}.
	\]
	On the other hand, the stuck time $S_1^\Or$ of the original sampler is random, which introduces additional variance.
	To formalize this, we denote the contributions from the unstuck ``slab'' and stuck ``spike'' parts
	as $Z_\mathrm{slb} = \int_{\{t: x_t \neq 0\}}f(x_t)dt - \mu (T_1^\Or - S_1^\Or) $ and $Z_\mathrm{spk} = S_1^\Or\{f(0) - \mu\}$, which are independent of each other because of the sampler's Markovian property.
	We now lower bound the original sampler's asymptotic variance by the latent sampler's through the variance decomposition formula:
	\begin{align*}
		&\hspace*{-1.75em}\frac{1}{\tau}\var\left[\int_{\{t: x_t \neq 0\}}f(x_t)\diff t - \mu (T_1^\Or - S_1^\Or) + S_1^\Or\{f(0) - \mu\}\right]\\
		&=\frac{1}{\tau}\var\left(Z_\mathrm{slb} + Z_\mathrm{spk} \right)\\
		&\geq \frac{1}{\tau}\var \left\{\E(Z_\mathrm{slb} +  Z_\mathrm{spk} \mid  Z_\mathrm{slb}) \right\} \\
		&= \frac{1}{\tau}\var \left\{Z_\mathrm{slb} +  \E(Z_\mathrm{spk} \mid  Z_\mathrm{slb}) \right\} \\
		&= \frac{1}{\tau}\var \left\{Z_\mathrm{slb} +  \E(Z_\mathrm{spk}) \right\} \\
		&= \frac{1}{\tau}\var \left(Z_\mathrm{slb} \right) \\
		&=\nu^\La.
	\end{align*}
\end{proof}

\section{Changing latent universe width}\label{supp:changing_width}

\begin{figure}[h]
	\centering
	\vspace{2mm}
	\includegraphics[width=10cm]{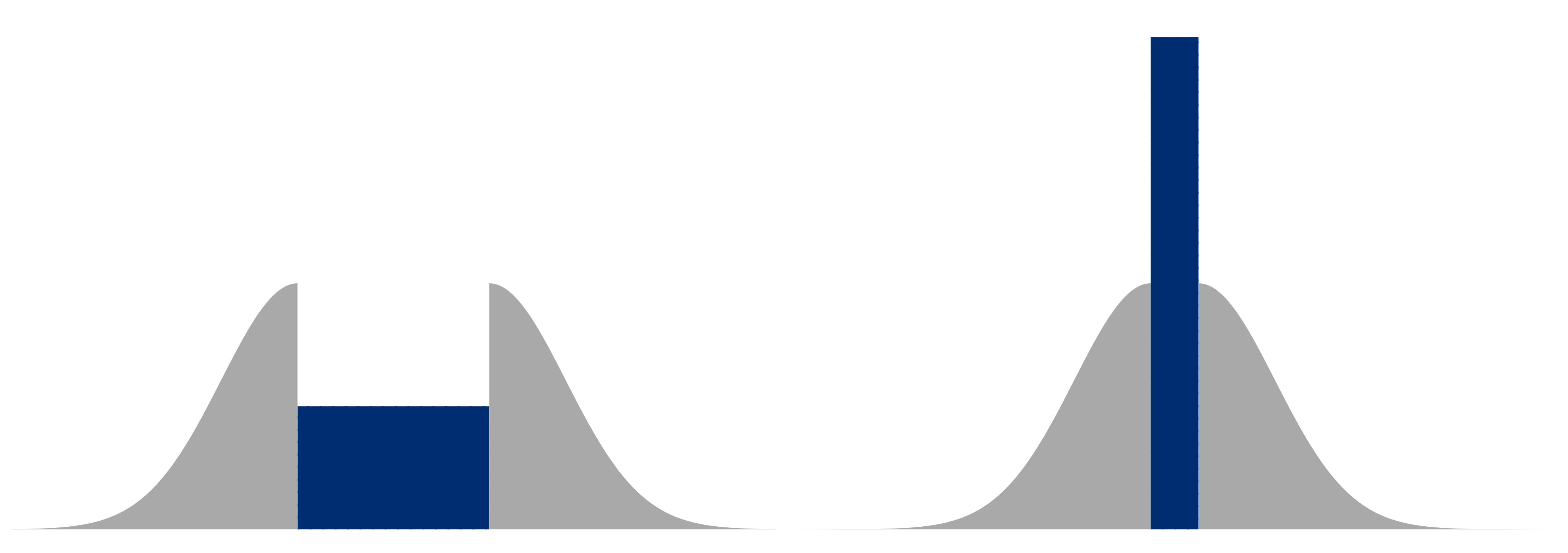}
	\caption{Left: Making the universe wider leads to multimodal posteriors, and samplers can bounce instead of entering the universe. Right: Making the universe narrower means samplers are guaranteed to enter the universe, but may bounce within the universe instead of exiting, leading to random exit directions.}
	\label{fig:alt_widths}
\end{figure}

Our choice of the latent universe width $w$ as in Equation~\eqref{eq:unstick_rate} is not the only one, but is likely an optimal one in the following sense.
Suppose we use a new width $c w$, scaling the original width by a factor $c \neq 1$.
Now, the density height needs to be adjusted accordingly to be $c^{-1} \slab(0)$ to maintain the same amount of mass in the latent universe.
This results in a discontinuous latent density and the latent sticky sampler now has to incorporate reflective behavior when trying to enter and exit the latent universe: 
the entry move is ``accepted'' with probability $\min\{1, c^{-1}\}$ and the exit move with probability $\min\{1, c\}$ \citep{chevallier2024pdmp}.
Otherwise, the sampler's coordinate fails to cross the discontinuity boundary and bounces with a flip in the coordinate's velocity.

For $c > 1$, the increased width and decreased height make the latent density multimodal (Figure \ref{fig:alt_widths}), undermining the sampler's ability to enter the universe and cross over to the other side.
This hampers efficient exploration.
For $c < 1$, the decreased width and increased height means the sampler is guaranteed to enter the universe, but only successfully exits from the other side with probability $c < 1$. 
If it fails, the sampler bounces and travels back to the other side and tries to exit again;
in this manner, the sampler keeps trying to exit from either side of the universe until it succeeds.
This creates randomness in its exit direction and causes some wasteful backtracking of the trajectory to the same side as the sampler entered.

For the case of decreasing width $c \to 0$, we can in fact show that the latent sampler converges to a version of the reversible jump \pdmp{} of \citet{chevallier2023reversible}, whose diffusive behavior has been pointed out as undesirable by \citet{bierkens2023sticky}.
The total number of exit attempts in a unit interval equals $\lfloor 1/(cw) \rfloor$, each of which has success probability of $c$; 
as $c\to0$, therefore, the time to exit converges to an exponential random variable of rate $w^{-1}$. 
In other words, we obtain a process which sticks for the same exponential amount of time as the original sticky sampler, but now with a random exit velocity.

\section{Adjusted Numerical Zig-Zag}\label{supp:adjusted_zz}

For the simulation study on the logistic regression case (Section~\ref{subsec:logreg_setup} and \ref{subsec:logreg_results}), we explored the numerical approximations of the zig-zag process on the latent density both with and without the non-reversible Metropolis adjustment, using Algorithms~1 and 3 of \citet{bertazzi2025piecewise}.
Their integrator starts with a half step in the position parameter, applies random coordinate-wise flips in the velocity according to Poisson event rates, and takes another half step in the position. 
The adjusted algorithm then concludes with a non-reversible Metropolis accept-reject step; 
if rejected, the sampler retains the initial position and reverses the initial velocity across \textit{all} coordinates.

We found this velocity reversal upon rejection to have a major negative impact on the resulting sampler's mixing when dealing with the latent density.
A dynamics-based sampler ordinarily has no problem traversing the latent universe along each coordinate because the density there has the partial derivative equal to zero;
in particular, both the analytic and unadjusted zig-zags maintain persistent motion within the latent universe and never backtrack.
However, the non-reversible Metropolis adjustment causes the adjusted zigzag to reverse its direction along a coordinate within the latent universe if there is a rejection driven by approximation errors along other coordinates.
This leads to random-walk behavior and inefficient exploration which, notably, cannot necessarily be remedied by tuning the step size to achieve a high acceptance rate;
as the accept-reject step is carried out at each numerical integration step, the sampler can still undergo a rejection and velocity reversal before it successfully exits the latent universe from the other side.
This is illustrated in Figure~\ref{fig:adj_vs_unadj}, which shows the trace plot along a selected coordinate in both the original and latent spaces.
While the adjusted sampler uses a step size that yields an acceptance rate above 99\%, it still undergoes thousands of rejections over the course of the sampling run, resulting in an order of magnitude loss in efficiency.

\begin{figure}[h]
	\centering
	\vspace{2mm}
	\includegraphics[width=14cm]{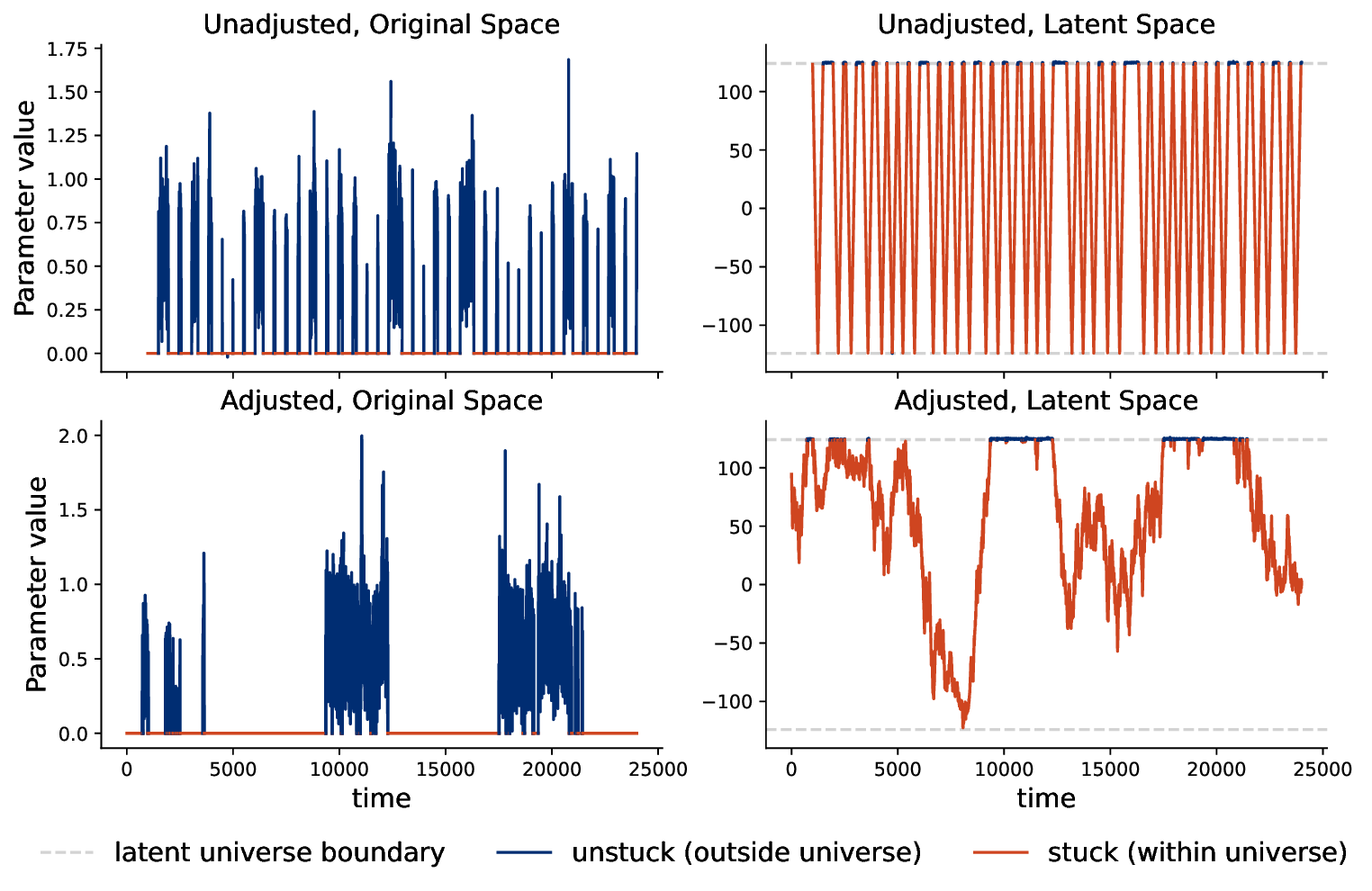}
	\caption{%
		Trace plots along a selected coordinate, when sampling from the logistic regression posterior of Section~\ref{subsec:logreg_setup} with $p_\mathrm{slab} = 0.01$, for both the adjusted (upper) and unadjusted (lower) numerically integrated zig-zags.
		The trace plots in the original space is shown on the left and in the latent space on the right, with the the width of the latent universe $w=248.16$ per Equation~\eqref{eq:unstick_rate}.
		Segments of the trace plots have been colored based on whether the sampler is stuck (\textcolor{jhured}{red}) or unstuck (\textcolor{jhublue}{blue}).
		The adjusted sampler mixes poorly in the original space, failing to move efficiently between the spike and slab due to velocity flips within the latent universe caused by the non-reversible Metropolis adjustment.
	}%
	\label{fig:adj_vs_unadj}
\end{figure}
\FloatBarrier

\end{appendix}
\bibliographystyle{imsart-nameyear} 
\bibliography{ref}       

\end{document}